\documentclass[11pt]{article}

\usepackage{fullpage}

\usepackage{times}
\usepackage{soul}
\usepackage{url}
\usepackage[hidelinks]{hyperref}
\usepackage[utf8]{inputenc}
\usepackage[small]{caption}
\usepackage{graphicx}
\usepackage{amsmath}
\usepackage{amsthm}
\usepackage{booktabs}
\usepackage{algorithm}
\usepackage{algorithmic}
\usepackage{natbib}

\usepackage[utf8]{inputenc}

\usepackage{color}
\usepackage{amsmath}
\usepackage{amssymb}
\usepackage{nicefrac}
\usepackage{amsthm}
\usepackage{booktabs}
\usepackage{paralist}
\usepackage{xcolor}
\usepackage[inline]{enumitem}

\usepackage{subcaption}
\usepackage{todonotes}
\usepackage{pgfplots}
\usepackage{tikz}

\newtheorem{theorem}{Theorem}[section]
\newtheorem{proposition}[theorem]{Proposition}

\newtheorem{corollary}[theorem]{Corollary}
\newtheorem{definition}{Definition}[section]

\newcommand{\score}{\textrm{score}}

\newcommand{\sharpp}{{{\mathrm{\#P}}}}
\newcommand{\np}{{{\mathrm{NP}}}}
\newcommand{\fpt}{{{\mathrm{FPT}}}}
\newcommand{\p}{{{\mathrm{P}}}}

\newcommand{\calS}{{\mathcal{S}}}

\definecolor{darkgreen}{rgb}{0,0.5,0}
\definecolor{darkpink}{rgb}{0.75,0.25,0.25}
\definecolor{RED}{rgb}{1,0,0}

\newcommand{\figurecut}[1]{}
\newcommand{\cutfornow}[1]{}
\newcommand{\conp}{\ensuremath{{\mathrm{coNP}}}}
\newcommand{\us}{\ensuremath{{\mathrm{US}}}}

\newcommand{\wone}{{{\mathrm{W[1]}}}}

\newcommand{\sharpwone}{{{\mathrm{\#W[1]}}}}

\newcommand{\av}{\textsc{AV}}
\newcommand{\sav}{\textsc{SAV}}
\newcommand{\ccav}{\textsc{CCAV}}

\newcommand{\pav}{\textsc{PAV}}

\newcommand{\phragmen}{\textsc{Phragm{\'e}n}}

\newcommand{\sharpmatching}{\textsc{\#Matching}}
\newcommand{\uniquecommittee}{\textsc{Unique-Committee}}

\newcommand{\uniquethresholdcommittee}{\textsc{Unique-Threshold-Committee}}

\usepackage{environ}

\newcommand{\repeatproposition}[1]{  \begingroup
  \renewcommand{\theproposition}{\ref{#1}}  \expandafter\expandafter\expandafter\proposition
  \csname repproposition@#1\endcsname
  \endproposition
  \endgroup
  \setcounter{theorem}{\value{theorem}-1}
}

\NewEnviron{repproposition}[1]{  \global\expandafter\xdef\csname repproposition@#1\endcsname{    \unexpanded\expandafter{\BODY}  }  \expandafter\proposition\BODY\unskip\label{#1}\endproposition
}

\newcommand{\repeattheorem}[1]{  \begingroup
  \renewcommand{\thetheorem}{\ref{#1}}  \expandafter\expandafter\expandafter\theorem
  \csname reptheorem@#1\endcsname
  \endtheorem
  \endgroup
  \setcounter{theorem}{\value{theorem}-1}
}

\NewEnviron{reptheorem}[1]{  \global\expandafter\xdef\csname reptheorem@#1\endcsname{    \unexpanded\expandafter{\BODY}  }  \expandafter\theorem\BODY\unskip\label{#1}\endtheorem
}

\title{Ties in Multiwinner Approval Voting}
\author{
  Łukasz Janeczko\\
  AGH University\\
  Kraków, Poland
  \and
  Piotr Faliszewski\\
  AGH University\\
  Kraków, Poland
}

\begin{document}

%%%%%%%%%%%%%%%%%%%%%%%%%%%%%%%%%%%%%%%%%%%%%%%%%%%%%%%%%%%%%%%%%%%%%%%%%

\maketitle

% Include a short abstract here (100-300 words):
\begin{abstract}
    We study the complexity of deciding whether there is a tie in a given
    approval-based multiwinner election, as well as the complexity of
    counting tied winning committees. We consider a family of Thiele
    rules, their greedy variants, Phragm{\'e}n's sequential rule, and
    Method of Equal Shares. For most cases, our problems are
    computationally hard, but for sequential rules we find an FPT
    algorithm for discovering ties (parameterized by the committee
    size). We also show experimentally that in elections of moderate
    size ties are quite frequent.
\end{abstract}

\section{Introduction}

In an approval-based multiwinner election, a group of voters expresses
their preferences about a set of candidates---i.e., each voter
indicates which of them he or she approves---and then, using some
prespecified rule, the organizer selects a winning committee (a
fixed-size subset of the candidates).
%Typically, the size of this committee is agreed upon prior to casting the votes. 
Multiwinner elections can be used to resolve very serious
matters---such as choosing a country's parliament---or rather
frivolous ones---such as choosing the tourist attractions that a group
of friends would visit---or those positioned anywhere in between these
two extremes---such as choosing a department's representation for the
university senate. In large elections, one typically does not expect
ties to occur (although surprisingly many such cases are
known\footnote{\texttt{https://en.wikipedia.org/wiki/List\_of\_close\_
    election\_results]}}), but for small and moderately sized ones the
issue is unclear.  While perhaps a group of friends may manage to not
spoil their holidays upon discovery that they were as willing to visit
one monument as another, a person not selected for a university senate
due to a tie may be quite upset, especially if this tie is discovered
after announcing the results.  To address such possibilities, we study
the following three issues:
\begin{enumerate}
\item We consider the complexity of detecting if two or more
  committees tie under a given voting rule. While for most rules this
  problem turns out to be intractable, for many settings we find
  practical solutions (in most cases it is either possible to use a
  natural integer linear programming trick or an FPT algorithm that we
  provide).
\item We consider the complexity of counting the number of winning
  committees. We do so, because being able to count winning committees
  would be helpful in sampling them uniformly. Unfortunately, in this
  case we mostly find hardness and hardness of approximation results.
\item We generate a number of elections, both synthetic and based on
  real-life data, and evaluate the frequency of ties. %how often the
  % rules that we
  % consider indeed
  % produce ties.
  It turns out to be surprisingly high.
  %Hence designing tie-breaking mechanisms is a pressing issue.
\end{enumerate}
% two main issues. First, we consider the complexity of detecting if
% there is a tie in a given multiwinner election (and, if there is one,
% the complexity of counting how many committees are tied). Second, we
% evaluate experimentally how often ties happen in elections of moderate
% size (say with tens of candidates and voters). We chose this size
% because we believe that, in practice, this is the setting where
% advanced multiwinner voting rules would be used most often.  Indeed,
% it is much more realistic that a university department would use the
% Phragm{\'e}n rule to elect its representatives for the senate than
% that a country would use this rule to elect its parliament; for a
% discussion of various approval-based multiwinner voting rules, see the
% overview of \citet{lac-sko:b:multiwinner-approval}.

We consider a subfamily of Thiele %multiwinner
rules
\citep{Thie95a,azi-gas-gud-mac-mat-wal:c:approval-multiwinner,lac-sko:c:approval-thiele}
that includes the multiwinner approval rule (AV), the approval-based
Chamberlin--Courant rule (CCAV), and the proportional approval voting
rule (PAV), as well as on their greedy variants.
% such as GreedyCCAV or
% GreedyPAV.
We also study satisfaction approval voting (SAV), the Phragm{\'e}n
rule, and Method of Equal Shares (MEqS).
This set includes rules appropriate for selecting committees of
individually excellent candidates (e.g., AV or SAV), diverse
committees (e.g., CCAV or GreedyCCAV), or proportional ones (e.g.,
PAV, GreedyPAV, Phragm{\'e}n, or MEqS); see the works of
\citet{elk-fal-sko-sli:j:multiwinner-properties} and
\citet{fal-sko-sli-tal:b:multiwinner-voting} for more details on
classifying multiwinner rules with respect to their application.  We
summarize our results in Table~\ref{tab:results}. See also the
textbook of \citet{lac-sko:b:multiwinner-approval}.

The issue of ties and tie-breaking has already received quite some
attention in the literature, although typically in the context of
single-winner voting.  For example,
\citet{obr-elk:c:random-ties-matter} and
\citet{obr-elk-haz:c:ties-matter} consider how various tie-breaking
mechanisms affect the complexity of manipulating elections, and
recently \citet{xia:c:probability-of-ties} has made a breakthrough in
studying the probability that ties occur in large, randomly-generated
single-winner elections.  \citet{xia:t:tie-breaking} also developed a
novel tie-breaking mechanisms, which can be used for some multiwinner
rules, but he did not deal with such approval rules as we study here.
Finally, \citet{con-rog-xia:c:mle} have shown that deciding if a
candidate is a tied winner in an STV election is $\np$-hard. While STV
is not an approval-based rule and they focused on the single-winner
setting, many of our results are in similar spirit.

\begin{table}[t]
  \centering
%  \scalebox{0.76}{
  \begin{tabular}{l|cc}
    \toprule
    Rule & \textsc{Unique-Committee}& \textsc{\#Winning-Committees} \\
    \midrule
    AV & $\p$ & $\p$ \\
    SAV & $\p$ & $\p$ \\
    \midrule
    CCAV & $\conp$-hard, $\mathrm{coW[1]}$-h. ($k$) &  $\sharpp$-hard, $\sharpwone$-hard ($k$)\\
    PAV  & $\conp$-hard, $\mathrm{coW[1]}$-h. ($k$) &  $\sharpp$-hard, $\sharpwone$-hard ($k$)\\
    \midrule
    GreedyCCAV  & $\conp$-com., $\fpt(k)$ &  $\sharpp$-hard, $\sharpwone$-hard ($k$)\\
    GreedyPAV   & $\conp$-com., $\fpt(k)$ &  $\sharpp$-hard, $\sharpwone$-hard ($k$)\\
    Phragm{\'e}n & $\conp$-com., $\fpt(k)$ &  $\sharpp$-hard, $\sharpwone$-hard ($k$)\\
    MEqS (Phase 1)        & $\conp$-com., $\fpt(k)$ &  $\sharpp$-hard\\
    \bottomrule
  \end{tabular}
%  }
  \caption{\label{tab:results}Summary of our complexity results.}
\end{table}

\section{Preliminaries}

By $\mathbb{R}_+$ we denote the set of nonnegative real numbers. For
each integer $t$, we write $[t]$ to mean $\{1, \ldots, t\}$.  We use
the Iverson bracket notation, i.e., for a logical expression $F$, we
interpret $[F]$ as~$1$ if $F$ is true and as~$0$ if it is false.  Given
a graph $G$, we write $V(G)$ to denote its set of vertices and $E(G)$
to denote its set of edges. For a vertex $v$, by $d(v)$ we mean its
degree (i.e., the number of edges that touch it).

An election $E = (C,V)$ consists of a set of candidates
$C = \{c_1, \ldots, c_m\}$ and a collection of voters
$V = (v_1, \ldots, v_n)$, where each voter $v_i$ has a set
$A(v_i) \subseteq C$ of candidates that he or she approves. We refer
to this set as~$v_i$'s approval set or $v_i$'s vote,
interchangeably. A multiwinner voting rule $f$ is a function that
given an election $E = (C,V)$ and committee size $k \in [|C|]$ outputs
a family of size-$k$ subsets of~$C$, i.e., a family of winning
committees. Below we describe the rules that we focus on.

Let $E = (C,V)$ be an election and let $k$ be the committee size.
Under the multiwinner approval rule ($\av$), each voter assigns a single
point to each candidate that he or she approves and winning committees
consist of $k$ candidates with the highest scores. Satisfaction
approval voting ($\sav$) proceeds analogously, except that each voter
$v \in V$ assigns $\nicefrac{1}{|A(v)|}$ points to each candidate he
or she approves. In other words, under AV each voter can give a single
point to each approved candidate, but under SAV he or she needs to
split a single point equally among them.

Next we consider the class of Thiele rules, defined originally by
\citet{Thie95a} and discussed, e.g., by
\citet{lac-sko:c:approval-thiele} and
\citet{azi-gas-gud-mac-mat-wal:c:approval-multiwinner}.  Given a
nondecreasing weight function
$w \colon \mathbb{N} \rightarrow \mathbb{R}_+$ such that $w(0) = 0$,
we define the $w$-Thiele score ($w$-$\score$) of a committee 
$S = \{s_1, \ldots, s_k\}$ in election $E$ to be:
\[
  \textstyle w\hbox{-}\score_E(S) = \sum_{v \in V} w(|A(v) \cap S|).
\]
The $w$-Thiele rule outputs all committees with the highest $w$-score.
We require that for each of our weight functions $w$, it is possible
to compute each value $w(i)$ in polynomial time with respect to~$i$. 
Additionally, we focus on functions such that $w(1) = 1$ and for
each positive integer $i$ it holds that
$w(i)-w(i-1) \geq w(i+1)-w(i)$. We refer to such functions, and the
Thiele rules that they define, as $1$-concave.
% \footnote{The  requirement that $w(1) = 1$ is a simple normalization.}
Three best-known 1-concave Thiele rules include the already defined AV
rule, which uses function $w_\av(t) = t$, the approval-based
Chamberlin--Courant rule ($\ccav$), which uses function
$w_\ccav(t) = [t \geq 1]$, and the proportional approval voting rule
($\pav$), which uses function
$w_\pav(t) = \sum_{i=1}^t\nicefrac{1}{i}$.
%Note that all these rules
%are $1$-concave.
%

%For brevity, instead of
%$\score_E^{w_\av}(S)$ we write $\score_E^\av(S)$ to refer to the score
%of committee $S$ under AV. We use analogous convention for the other
%Thiele rules.

While it is easy to compute some winning committee under the AV rule
in polynomial time (out of possibly exponentially many), for the other
Thiele rules, including CCAV and PAV, even deciding if a committee
with at least a given score exists is $\np$-hard (see the works of
\citet{pro-ros-zoh:j:proportional-representation} and
\citet{bet-sli-uhl:j:mon-cc} for the case of CCAV, and the works of
\citet{azi-gas-gud-mac-mat-wal:c:approval-multiwinner} and
\citet{sko-fal-lan:j:collective} for the general case).  Hence,
sometimes the following greedy variants of Thiele rules are used ($E$
is the input election and $k$ is the desired committee size):
\begin{enumerate}
\item[] Let $f$ be a $w$-Thiele rule. Its greedy variant, denoted
  Greedy-$f$, first sets $W_0 := \emptyset$ and then executes $k$
  iterations, where for each $i \in [k]$, in the $i$-th iteration it
  computes $W_i := W_{i-1} \cup \{c\}$ such that $c$ is a candidate in
  $C \setminus W_{i-1}$ that maximizes the $w$-score of $W_i$.
  Finally, it outputs $W_k$.  In case of internal ties, i.e., if at
  some iteration there is more than one candidate that the algorithm
  may choose, the algorithm outputs all committees that can be
  obtained for some way of resolving each of these ties. In other
  words, we use the parallel-universes tie-breaking
  model~\citep{con-rog-xia:c:mle}.
\end{enumerate}
When we discuss the operation of some Greedy-$f$ rule on election $E$
and we discuss the situation after its $i$-th iteration, where, so
far, subcommittee $W_i$ was selected, then by the score of a
(not-yet-selected) candidate~$c$ we mean the value
$w\hbox{-}\score_E(W_i \cup \{c\}) - w\hbox{-}\score_E(W_i)$, i.e.,
the marginal increase of the $w$-score that would result from
selecting $c$.
We refer to the greedy variants of CCAV and PAV as Greedy\-CCAV and
GreedyPAV (in the literature, these rules are also sometimes called
\emph{sequential} variants of CCAV and PAV, see, e.g., the book of
\citet{lac-sko:b:multiwinner-approval}).
% ; note that the greedy variant of AV is AV itself).
Given a greedy variant of a 1-concave Thiele
rule, it is always possible to compute at least one of its winning
committees in polynomial time by breaking internal ties arbitrarily.
Further, it is well-known that the $w$-score of this committee is at
least a $1-\nicefrac{1}{e} \approx 0.63$ fraction of the highest
possible $w$-score; this follows from the classic result of
\citet{nem-wol-fis:j:submodular} and the fact that $w$-score is
monotone and submodular.

% Finally, we consider the Phragm{\'e}n rule (sometimes referred to in
% the literature as

The Phragm{\'e}n (sequential) rule proceeds as follows  (see, e.g., the
work of \citet{san-elk-lac-fer-fis-bas-sko:c:pjr}):
% It is similar to
% the greedy ones in that it also selects the members of the committee
% one by one:
\begin{enumerate}
\item[] Let $E = (C,V)$ be an election and let $k$ be the committee
  size.  Each candidate costs a unit of currency. The voters start with no
  money, but they receive it continuously at a constant rate. As soon
  as there is a group of voters who approve a certain not-yet-selected
  candidate and who together have a unit of currency, these voters ``buy''
  this candidate (i.e., they give away all their money and the
  candidate is included in the committee). The process stops as soon
  as $k$ candidates are selected. For internal ties, we use the
  parallel-universes tie-breaking.
  % (i.e., if a committee could be
  % selected for some way of breaking internal ties, then it is among
  % the winning committees).
\end{enumerate}
Method of Equal Shares (MEqS), introduced by \citet{pet-sko:laminar}
and \citet{pet-pie-sko:c:meqs}, is similar in spirit, but gives the
voters their ``money'' up front (we use the same notation as above):
\begin{enumerate}
\item[] Initially, each voter has budget equal to
  $\nicefrac{k}{|V|}$. The rule starts with an empty committee and
  executes up to $k$ iterations as follows (for each voter $v$, let
  $b(v)$ denote $v$'s budget in the current iteration): For each
  not-yet-selected candidate $c$ we check if the voters that approve
  $c$ have at least a unit of currency (i.e.,
  $\sum_{v \in A(c)} b(v) \geq 1$). If so, then we compute value
  $\rho_c$ such that $\sum_{v \in A(c)} \min(b(v), \rho_c) = 1$, which
  we call the per-voter cost of $c$.  We extend the committee with
  this candidate $c'$, whose per-voter cost $\rho_{c'}$ is lowest; the
  voters approving $c'$ ``pay'' for him or her (i.e., each voter
  $v \in A(c')$ gives away $\min(b(v), \rho_{c'})$ of his or her
  budget). In case of internal ties, we use the parallel-universes
  tie-breaking.  The process stops as soon as no candidate can be
  selected.
\end{enumerate}
The above process, referred to as Phase~1 of MEqS, often selects fewer
than $k$ candidates. To deal with this, we extend the committee with
candidates selected by Phragm{\'e}n (started off with the budgets that
the voters had at the end of Phase~1).  We jointly refer to the greedy
rules, Phragm{\'e}n, MEqS, and Phase~1 of MEqS as sequential rules.

We assume that the reader is familiar with basic classes of
computational complexity such as $\p$, $\np$, and $\conp$.  $\sharpp$
is the class of functions that can be expressed as counting accepting
paths of nondeterministic polynomial-time Turing
machines. Additionally, we consider pararameterized complexity classes
such as $\fpt$ and $\wone$. $\sharpwone$ is a parameterized counting
class which relates to $\wone$ in the same way as $\sharpp$ relates to
$\np$~\citep{flu-gro:j:parameterized-counting}.  When discussing
counting problems, it is standard to use Turing reductions: A counting
problem $\#A$ reduces to a counting problem $\#B$ if there is a
polynomial time algorithm that solves $\#A$ in polynomial time,
provided it has oracle access to $\#B$ (i.e., it can solve $\#B$ in
constant time).\footnote{For $\#\wone$, the running time can even be
  larger, but our $\#\wone$-hardness proofs use polynomial-time
  reductions.}
%

% For example, the problem of counting independent sets of size $k$ in a
% graph is both $\sharpp$-complete and $\#\wone$-complete for the
% parameterization by $k$.

% Finally, we consider the class $\us$ (\textsc{Unique Solutions}) due
% to \citet{bla-gur:j:unique-sat}.  A decision problem is in $\us$ if
% there is a nondeterministic polynomial-time Turing machine that on
% every \emph{yes}-instance of the problem has exactly one accepting
% computation path and on every \emph{no}-instance has either zero or at
% least two accepting paths.  $\us$ contains $\conp$ and is included in
% $\dpclass$ \cite{bla-gur:j:unique-sat}.
% %
% $\uniqueindependentset$, i.e., the problem of testing if a graph has
% exactly one independent set of at least a given size, is
% $\us$-complete~\citep{hudry-lobstein:unique-vertex-cover}.

\section{Unique Winning Committee}\label{sec:unique}

In this section we consider the problem of deciding if a given
multiwinner rule outputs a unique committee in a given election.
Formally, we are interested in the following problem.

\begin{definition}
  Let $f$ be a multiwinner voting rule.  In the
  $f$-$\uniquecommittee$ problem we are given an election~$E$
  and a committee size $k$, and we ask if $|f(E,k)| = 1$.
\end{definition}

It is a folk result that for AV and SAV this problem is in $\p$ (see
beginning of Section~\ref{sec:counting} for an argument).
% Indeed, for
% both these rules, given an election $E$ and committee size $k$, it
% suffices to sort the candidates in the nonincreasing order of their
% scores and check if the $k$-th one has higher score than than the
% $(k+1)$-st one. If so, there is a unique winning committee and
% otherwise---if the scores of the $k$-th and $(k+1)$-st candidates are
% equal---there are at least two winning committees.
%
%\begin{proposition}
%  $\{$AV, SAV$\}$-$\uniquecommittee$ $\in \p$.
%\end{proposition}
%
For Thiele rules other than AV, the situation is more intriguing.  In
particular, already the problem of deciding if a given committee is
winning under the CCAV rule is
$\conp$-complete~\citep{son-dey-mis:c:multiwinner-verification}. We
show that for 1-concave Thiele rules other than AV the
$\uniquecommittee$ problem is $\conp$-hard (and we conjecture that the
problem is not in $\conp$).

%  \ljnote{I think that we do not need an assumption about $1$-concave because each voters approves exactly 2 voters (or, $i-1 + 2$ in general with $i: $ w(i)-w(i-1) > w(i+1)-w(i)). We can generalize it to any reasonable Thiele rule,
%    that is, $\exists$ $i \geq 1: w_R(i)-w_R(i-1) > w_R(i+1)-w_R(i)$}
\begin{proposition}\label{pro:thiele-unique}
  Let $f$ be a 1-concave $w$-Thiele rule other than AV. Then
  $f$-$\uniquecommittee$ is $\conp$-hard.
\end{proposition}
\begin{proof}
  Let $x = w(2) - w(1)$ and assume, for now, that $x < 1$. We give a
  reduction from \textsc{Independent-Set} to the complement of
  $f$-$\uniquecommittee$.  An instance of \textsc{Independent-Set}
  consists of a graph $G$ and integer $k$, and we ask if there are $k$
  vertices neither of which is connected with the others. Let $G'$ be a graph
  obtained from $G$ by adding $k$ vertices such that each of the new
  vertices is connected to each of the old ones (but the new vertices
  are not connected to each other). If $G$ does not have a size-$k$
  independent set, then $G'$ has a unique one, and if $G$ has at least
  one size-$k$ independent set, then $G'$ has at least two.  Let us
  denote the vertices of $G'$ as $V(G') = \{v_1, \ldots, v_n\}$ and
  its edges as $E(G') = \{e_1, \ldots, e_m\}$.  Let $\delta$ be the
  highest degree of a vertex in $V(G')$. We fix the committee size to
  be $k$ and we form an election $E$ with candidate set $V(G')$ and
  with the following voters:
  \begin{enumerate}
  \item For each edge $e_\ell = \{v_i,v_j\}$ there is a single voter
    who approves $v_i$ and $v_j$. 
  \item For each vertex $v_i$ there are $\delta-d(v_i)$ voters
    approving~$v_i$.
  \end{enumerate}
  Consider a set of $k$ vertices from $V(G')$. If this set is an
  independent set, then interpreted as a committee in election~$E$, it
  has $w$-score equal to $\delta k$.  On the other hand, if $S$ is not
  an independent set, then its score is at most
  $(\delta k-1)+x < \delta k$.
  We know that $G'$ has an independent set of size
  $k$.  If $G$ also has one, then our election has at least two
  winning committees and, otherwise, the winning committee is unique.
  
  Let us now consider the case that $x= 1$.  Since $f$ is not AV, there
  certainly is an integer $t$ such that $w(t)-w(t-1) = 1$ and
  $w(t+1)-w(t) < 1$.  In this case, we modify the reduction by adding
  $t-1$ candidates approved by every voter and changing the committee
  size to be $t+k-1$.
\end{proof}

For greedy variants of Thiele rules (with the natural exception of AV)
and for the Phragm{\'e}n rule, deciding if the winning committee is
unique is $\conp$-complete. Our proof for the greedy variants of
Thiele rules is inspired by a complexity-of-robustness proof for
GreedyPAV, provided by
\citet{fal-gaw-kus:c:greedy-approval-robustness}. For Phragm{\'e}n,
somewhat surprisingly, their robustness proof directly implies our
desired result. We also get analogous result for Method of Equal
Shares and its Phase~1.

\begin{theorem}\label{thm:uniqgreedyhard}
  Let $f$ be a 1-concave $w$-Thiele rule, $f \neq$ AV.
  Greedy-$f$-\textsc{Unique-Comm\-ittee} is $\conp$-complete.
\end{theorem}
\begin{proof}
    Membership in $\conp$ is clear: Given an election and committee size,
  we run the greedy algorithm breaking the ties arbitrarily, and we
  compute some winning committee $W$. Then, we rerun the same
  algorithm nondeterministically, at each internal tie trying each
  possible choice; if a given computation completes with a committee
  different than $W$ then it rejects (and the whole computation
  rejects; indeed, we found two different winning committees) and
  otherwise it accepts (if all paths accept, then the whole
  computation accepts; indeed, all ways of handling the internal ties
  lead to the same final committee). In the following, we focus on
  showing $\conp$-hardness.\smallskip
  
  Let $\delta_1 = w(1)-w(0)$, $\delta_2 = w(2)-w(1)$, and $\delta_3 = w(3)-w(2)$.
  For example, for $w_\pav$ we would have $\delta_1 = 1$, $\delta_2 = \frac{1}{2}$ and $\delta_3 = \frac{1}{3}$. By our assumptions
  on weight functions, we know that (a)~these numbers are rational, 
  (b)~$\delta_1 = 1$ (but we
  will not use this), and that 
  (c)~$\delta_1 \geq \delta_2 \geq \delta_3$. We additionally assume that
  $\delta_1 - \delta_2 > \delta_2 - \delta_3$, but later we will show how to relax this assumption.

  We give a reduction from \textsc{Independent Set} to the complement
  of Greedy-$f$-$\uniquecommittee$. Our input consists of a
  graph $G$, where $V(G) = \{v_1, \ldots, v_n\}$ and
  $E(G) = \{e_1, \ldots, e_m\}$, and an integer $k$. The question is if
  there are $k$ vertices in $V(G)$ that are not connected by an edge.
  Without loss of generality, we assume that $G$ is $3$-regular, i.e.,
  each vertex touches exactly three edges~\citep{gar-joh-sto:j:simplified-np-complete}.
  Let~$\alpha$ be a positive integer such that $\alpha \delta_1$ and
  $\alpha\frac{\delta_1-\delta_2}{\delta_1}$ are integers and
  $\alpha (\delta_1-\delta_2) > \delta_1$. We fix values
  $t = \alpha (nmk)^3$, $T = 10\alpha(nmk)^6$, and $D=\beta T^{10}$, where $\beta$ is the smallest
  positive integer greater than $\frac{\delta_1}{\delta_1-\delta_2}$; while we could choose
  smaller ones, these suffice.

  We form an election where the candidate set is $V(G) \cup \{p,d\}$ and
  we have the following three groups of voters:
  \begin{enumerate}
  \item For each edge $e_\ell = \{v_i,v_j\}$, we have $t$ voters
    with approval set $\{v_i,v_j,d\}$.
  \item For each candidate $v_i$ we have $D + T^3 + ((m-1)n+1)T-3t$ voters with approval
    set $\{v_i\}$, and for each pair of distinct candidates $v_i$ and
    $v_j$ we have $T$ voters with approval set $\{v_i,v_j\}$.
  \item We have $D+T^3 + nmT + 1 - \frac{(\delta_1-\delta_2)}{\delta_1}kT -mt$ voters with
    approval set $\{p,d\}$, $\frac{3(\delta_1-\delta_2)}{\delta_1}kt$ voters who
    approve $d$, and $mt$ voters who approve $p$.
  \end{enumerate}
  We let the committee size be $n+1$. We claim that if~$G$ contains an
  independent set of size~$k$ then there are two Greedy-$f$ winning
  committees, $V(G) \cup \{d\}$ and $V(G) \cup \{p\}$, and otherwise
  there is only one, $V(G) \cup \{d\}$. The proof follows.

  Let $X =\delta_1D + \delta_1T^3 + \delta_1nmT$.  Prior to the first iteration of Greedy-$f$, each
  candidate $v_i$ has score $X$, candidate $d$ has score:
  %\begin{align*}
  $
%    (\delta_1T^3 &+ \delta_1nmT + \delta_1 -(\delta_1-\delta_2)kT) + 3(\delta_1-\delta_2)kt \\ &
    X - (\delta_1-\delta_2)kT + 3(\delta_1-\delta_2)kt + \delta_1. 
  $
%  \end{align*}
  and candidate $p$ has score:
%  \begin{align*}
  $
%    \delta_1T^3 + \delta_1nmT + \delta_1 -(\delta_1-\delta_2)kT =
    X - (\delta_1-\delta_2)kT + \delta_1.
  $
%  \end{align*}
  During the first $k$ iterations, Greedy-$f$ selects some $k$
  candidates from $V(G)$.  This is so, because whenever some
  candidate~$v_i$ is selected, the scores of the remaining members
  of~$V(G)$ decrease by $(\delta_1-\delta_2)T$ due to the voters in the second
  group, and by at most $(\delta_1-\delta_2)t$, due to the voters in the first
  group.  Hence, after the first $k-1$ iterations each remaining
  candidate from $V(G)$ has score at least
  $
    X - (\delta_1-\delta_2)(k-1)T -(\delta_1-\delta_2)(k-1)t,
  $
  which---by our choices of $\alpha$, $t$, and $T$---is larger than
  the scores that both $p$ and $d$ had even prior to the first
  iteration (note that the scores of the candidates cannot increase
  between iterations). On the other hand, after the $k$-th iteration,
  each remaining member of $V(G)$ has score at most $X - (\delta_1-\delta_2)kT$,
  which is less than $p$ has (since $p$ is only approved by voters who
  do not approve members of $V(G)$, at this point his or her score is the
  same as prior to the first iteration).  As a consequence, in the
  $(k+1)$-st iteration Greedy-$f$ either chooses $p$ or $d$. Let us now
  analyze which one of them.

  Let $S$ be the set of candidates from $V(G)$ selected in the first~$k$
  iterations. If $S$ forms an independent set, then prior to the
  $(k+1)$-st iteration, the score of $d$ is $X - (\delta_1-\delta_2)kT +
  \delta_1$. This is so, because for each candidate $v_i$ in $S$, $d$ loses
  exactly $3(\delta_1-\delta_2)t$ points due to the voters in the first group
  that correspond to the three edges that include $v_i$ (since $S$ is
  an independent set, for each member of $S$ these are different three
  edges). In this case, Greedy-$f$ is free to choose either among $p$
  and $d$. However, if $S$ is not an independent set, then the score
  of $d$ drops by at most $(3k-1)(\delta_1-\delta_2)t + (\delta_2-\delta_3)t$. This is so,
  because $S$ contains at least two candidates $v_i$ and $v_j$ that
  are connected by an edge; when the second one of them is included in
  the committee, then the score of $d$ drops by at most
  $2(\delta_1-\delta_2)+(\delta_2-\delta_3)$. In this case Greedy-$f$ is forced to select
  $d$ in the $(k+1)$-st iteration.  In the following $n-k$ iterations,
  $f$ selects the remaining members of $V(G)$ (after either $p$ or $d$ is
  selected in the $(k+1)$-st iteration, the score of the other one
  drops so much that he or she cannot be selected; this is due to the $D$ voters who approve $\{p,d\}$).

  % It remains to observe that if $G$ contains an independent set of
  % size~$k$, then Greedy-$f$ can choose its members in the first~$k$
  % iterations. This argument (and the remaining cases to be covered) are
  % in Appendix~\ref{app:greedy}.
%
  %\smallskip
%
%  The case where $\delta_1 - \delta_2 \leq \delta_2-\delta_3$ is
%  in Appendix~\ref{app:greedy}
%
  % Let us now consider the case where
  % $\delta_1 - \delta_2 \leq \delta_2-\delta_3$.  Let
  % $\delta_4 = w(4)-w(3)$, $\delta_5 = w(5)-w(4)$, and so on.  If there
  % is some positive integer $t$ such that
  % $\delta_{t+1} - \delta_{t+2} > \delta_{t+2} - \delta_{t+3}$ then it
  % suffices to use the same reduction as above, extended so that we
  % have candidates $d_1, \ldots, d_t$ that are approved by every voter
  % and the committee size is increased by $t$. Greedy-$f$ will choose
  % these~$t$ candidates in the first~$t$ iterations and then it will
  % continue as described in the reduction, with
  % $\delta_{t+1}, \delta_{t+2}$, and $\delta_{t+3}$ taking the roles of
  % $\delta_1$, $\delta_2$, and $\delta_3$.  In fact, such a $t$ must
  % exist. Otherwise, if
  % $\delta_{t+1} - \delta_{t+2} \leq \delta_{t+2} - \delta_{t+3}$ for
  % every $t$ then either $f$ is AV (which we assumed not to be the
  % case) or $w$ is not nondecreasing, which is forbidden by definition.

%%%%  

  It remains to observe that if $G$ contains an independent set of
  size~$k$, then Greedy-$f$ can choose its members in the first~$k$
  iterations.  This is the case, because whenever Greedy-$f$ chooses a
  member of the independent set, then the score of its other members
  never drops more than the score of the other remaining vertex
  candidates. Hence, if~$G$ has a size-$k$ independent set, then, due
  to the parallel-universes tie-breaking, Greedy-$f$ outputs two
  winning committees, $V(G) \cup \{p\}$ and $V(G) \cup
  \{d\}$. Otherwise we have a unique winning committee
  $V(G) \cup \{d\}$.  This completes the proof for the case that
  $\delta_1 - \delta_2 > \delta_2-\delta_3$.

%The case where $\delta_1 - \delta_2 \leq \delta_2-\delta_3$ is in
%Appendix~\ref{app:greedy}

Let us now consider the case where
$\delta_1 - \delta_2 \leq \delta_2-\delta_3$.  Let
$\delta_4 = w(4)-w(3)$, $\delta_5 = w(5)-w(4)$, and so on.  If there
is some positive integer $t$ such that
$\delta_{t+1} - \delta_{t+2} > \delta_{t+2} - \delta_{t+3}$ then it
suffices to use the same reduction as above, extended so that we have
candidates $d_1, \ldots, d_t$ that are approved by every voter and the
committee size is increased by $t$. Greedy-$f$ will choose these~$t$
candidates in the first~$t$ iterations and then it will continue as
described in the reduction, with $\delta_{t+1}, \delta_{t+2}$, and
$\delta_{t+3}$ taking the roles of $\delta_1$, $\delta_2$, and
$\delta_3$.  In fact, such a $t$ must exist. Otherwise, if
$\delta_{t+1} - \delta_{t+2} \leq \delta_{t+2} - \delta_{t+3}$ for
every $t$ then either $f$ is AV (which we assumed not to be the case)
or $w$ is not nondecreasing, which is forbidden by definition.
\end{proof}

\begin{corollary}\label{cor:cc-pav-phragmen}
  $\uniquecommittee$ is $\conp$-complete for GreedyCCAV,
  GreedyPAV, and $\phragmen$.
\end{corollary}

The results for GreedyCCAV and GreedyPAV follow directly from the
preceding theorem. For Phragm{\'e}n,
\citet{fal-gaw-kus:c:greedy-approval-robustness} have shown that the
following problem, known as
Phragm{\'e}n-\textsc{Add-Robustness-Radius}, is $\np$-complete: Given
an election~$E$, committee size~$k$, and number~$B$, is it possible to
add at most $B$ approvals to the votes so that the winning committee
under the resolute variant of the Phragm{\'e}n rule (where all
internal ties are resolved according to a given tie-breaking order)
changes. Their proof works in such a way that adding approvals only
affects how ties are broken. Hence, effectively, it also shows that
$\uniquecommittee$ is $\conp$-complete for the (non-resolute) variant
of Phragm{\'e}n.

\begin{theorem}\label{thm:meqs-p1-uni}
  $\uniquecommittee$ is $\conp$-complete for Phase~1 of MEqS.
\end{theorem}
\begin{proof}
  The following nondeterministic algorithm shows membership in
  $\conp$: First, we deterministically compute the output of Phase~1,
  breaking internal ties in some arbitrary way. This way we obtain
  some committee $W$.  Next we rerun Phase~1, at each internal tie
  nondeterministically trying all possibilities. We accept on
  computation paths that output~$W$ and we reject on those outputting
  some other committee. This algorithm accepts on all computation
  paths if and only if the rule has a unique winning committee.

  Next, we give a reduction from the complement of the classic
  $\np$-complete problem, \textsc{X3C}.  An instance of \textsc{X3C}
  consists of a universe set $U = \{u_1, \ldots, u_{3n}\}$ and a
  family $\calS = \{S_1, \ldots, S_{3n}\}$ of size-$3$ subsets of
  $U$. We ask if there are $n$ sets from $\calS$ whose union is $U$
  (we refer to such a family as an exact cover of $U$; note that the
  sets in such a cover must be disjoint). Without loss of generality,
  we assume that each member of $U$ belongs to exactly three sets from
  $\calS$~\citep{gon:j:x3c} and that $n$ is even.

  Now we describe our election. Ideally, we would like to distribute
  different amounts of budget between different voters, but as MEqS
  splits the budget evenly, we design the election in such a way that
  in the initial iterations the respective voters spend appropriate
  amounts of money on the candidates that otherwise are not crucial
  for the construction. We form the following groups of voters (we
  reassure the reader that the analysis is more pleasant than the
  following two enumerations may suggest):
  \begin{enumerate}
  \item Group $B$, which contains  $144n^3 - 12n$ voters.
  \item Group $B_U$, which contains
    $54n^3 + 9n^2$ voters.
  \item Group $U'$, which models the elements of the universe set~$U$.
    For each $u_i \in U$, there is a single corresponding voter in
    $U'$.  We have $|U'|=3n$.
  \item Group $U''$, which serves a similar purpose as $U'$, but
    contains more voters. Specifically, for each $u_i \in U$, there
    are $6n$ corresponding voters in $U''$;  $|U''| = 18n^2$.

  \item Group $V_{pd}$, which contains $12n$ voters.
  \item Group $V_S$, which contains $9n$ voters.
  \item Two voters, $d_1$ and $d_2$.
  \end{enumerate}
  In total, there are $198n^3 + 27n^2 + 12n + 2$ voters. Further, we
  have the following groups of candidates:

  \begin{enumerate}
  \item Group $C_B$ of $144n^3 - 12n^2$ candidates approved by the
    $144n^3$ voters from $B \cup V_{pd}$.
  \item Group $C_U$ of $54n^3 + 24n^2 + \nicefrac{5n}{2}$ candidates
    approved by the $54n^3 + 27n^2 + 3n$ voters from $B_U \cup U' \cup U''$.
    %, i.e., by    $54n^3 + 27n^2 + 3n$ voters.
  \item Candidate $p$ approved by the $12n$ voters from $V_{pd}$.

  \item Candidate $d$ approved by the $15n$ voters from
    $V_{pd} \cup U'$.
  \item Candidates $c_1$ and $c_2$, both approved by $d_1$ and $d_2$.
  \item Group $D$ of
    $15n^2 + \frac{45n}{2} + 5$ candidates approved by $d_1$.
  \item For each set $S_\ell \in \calS$ such that
    $S_\ell = \{u_i,u_j,u_t\}$ we have a corresponding candidate
    $s_\ell$ approved by: (a) three unique voters from $V_S$, (b) the
    voters from $U'$ and $U''$ that correspond to the elements $u_i$,
    $u_j$, $u_t$.  We write $S$ to denote this group of candidates and
    we refer to its members as the $S$-candidates. Each $S$-candidate is
    approved by $3 + 3 + 3 \cdot 6n = 18n + 6$ voters.
  \end{enumerate}

  We have $198n^3 + 27n^2 + 28n + 9$ candidates in total.  We set the
  committee size $k$ to be equal to the number of voters, i.e.,
  $k = 198n^3 + 27n^2 + 12n + 2$.
  % The committee is to be smaller than the number of candidates and,
  % hence, in principle ties are possible.
  Let us consider the following two committees (note that each of them
  contains fewer than $k$ candidates; indeed, Phase~1 of MEqS
  sometimes chooses committees smaller than requested):
  \begin{align*}
    W_d &= C_B \cup C_U \cup S \cup \{c_1,c_2\} \cup \{d\}, \\
    W_p &= C_B \cup C_U \cup S \cup \{c_1,c_2\} \cup \{p\}.
  \end{align*}
  We claim that Phase~1 of MEqS always outputs committee $W_d$, and if
  $(U,\calS$) is a \emph{yes}-instance then it also outputs $W_p$.

  Let us analyze how Phase~1 of MEqS proceeds on our election.  Since
  the committee size is equal to the number of voters, initially each
  voter receives budget equal to $1$.

  At first, we will select all candidates from $C_B$. Indeed, there
  are $144n^3-12n^2$ candidates in this group, each approved by
  $144n^3$ voters (from $B \cup V_{pd}$). Each of these voters pays
  $\nicefrac{1}{144n^3}$ for each of the candidates (this is the
  lowest per-voter candidate cost at this point).  After these
  purchases, each voter from $B \cup V_{pd}$ will be left with budget
  equal to
  $1 - (144n^3 - 12n^2) \cdot (\nicefrac{1}{144n^3}) = \nicefrac{1}{12n}$.

  Next, we will select all candidates from $C_U$.  Indeed, this set
  contains $54n^3+24n^2 + \nicefrac{5n}{2}$ candidates approved by
  $54n^3+27n^2+3n$ voters (from $B_U \cup U' \cup U''$) who have not
  spent any part of their budget yet. All candidates in $C_U$ will be
  purchased at the same pre-voter cost of
  $\nicefrac{1}{(54n^3+27n^2+3n)}$ (the lowest one at this
  point). Each voter in $B_U \cup U' \cup U''$ will be left with
  budget equal to
  $1 - (54n^3 + 24n^2 + \nicefrac{5n}{2}) \cdot \nicefrac{1}{(54n^3 +
    27n^2 + 3n)} = \frac{3n^2 + \nicefrac{n}{2}}{54n^3 + 27n^2 + 3n} =
  \frac{6n+1}{108n^2 + 54n + 6} = \frac{6n+1}{(6n+1) \cdot (18n+6)} =
  \nicefrac{1}{(18n+6)}$.

  Next, we consider the $S$-candidates who, at this point, have the
  highest approval score among the yet unselected candidates. As each
  $S$-candidate is approved by exactly $18n+6$ voters and each voter
  still has budget higher or equal to $\nicefrac{1}{(18n+6)}$, we keep
  selecting the $S$-candidates at the per-voter cost of
  $\nicefrac{1}{(18n+6)}$ as long as there is at least one such candidate
  whose all voters still have budget of at least
  $\nicefrac{1}{(18n+6)}$.

  Upon selecting a given $S$-candidate, corresponding to set $S_\ell$,
  all the voters who approve him or her pay
  $\nicefrac{1}{(18n+6)}$. This includes the three unique voters from
  $V_S$ and the voters from $U'$ and $U''$ who correspond to the
  members of $S_\ell$. Prior to this payment, the voters from $U'$ and
  $U''$ have budget equal to $\nicefrac{1}{(18n+6)}$, so they
  end up with $0$ afterward (and we say that they are \emph{covered}
  by this $S$-candidate).  Consequently, the $S$-candidates that we
  buy at the per-voter cost of $\nicefrac{1}{(18n+6)}$ correspond to
  disjoint sets.

  Now let us consider what happens when there is no $S$-candidate left
  who can be purchased at the per-voter cost of
  $\nicefrac{1}{(18n+6)}$. This means that for each unselected $S$
  candidate, at least $6n+1$ voters approving him have already been
  covered and have no budget left. Hence, for a given $S$-candidate
  there are at least $6n+1$ voters (from $U'$ and $U''$) whose budget
  is~$0$, at most $12n+2$ voters (from $U'$ and $U''$) who each have
  budget of $\nicefrac{1}{(18n+6)}$, and three voters (from $V_S$) who
  each have budget equal to $1$. To buy this $S$ candidate, the voters
  from $U'$ and $U''$ would have to use up their whole budget, and the
  voters from $V_S$ would have to pay at least:
  \[
    \textstyle \frac{1}{3}(1 - (12n+2) \cdot \frac{1}{18n+6}) = \frac{18n+6 -
      (12n+2)}{3 \cdot (18n+6)} = \frac{6n+4}{54n+18}
  \]
  each.  However, at this point there are two candidates that can be
  purchased at lower per-voter cost.

  Indeed, candidate $p$ could be purchased by the $12n$ voters from
  $V_{pd}$ at the per-voter cost of $\nicefrac{1}{12n}$ (after buying
  the candidates from $C_B$, they still have exactly this amount of
  budget left). Since candidate $d$ also is approved by all the voters
  from $V_{pd}$, and also by the voters from $U'$, candidate $d$ would
  either have the same per-voter cost as $p$ (in case all the members
  of $U'$ were already covered) or would have an even lower per-voter
  cost. The only other remaining candidates are $c_1$, $c_2$, and the
  candidates from $D$, but their per-voter costs are greater or equal
  to $\nicefrac{1}{2}$.  Hence, at this point, MEqS either selects $p$
  or $d$. The former is possible exactly if the already selected
  $S$-candidates form an exact cover of $U'$ (and, hence, correspond
  to an exact cover for our input instance of \textsc{X3C}).

  If we select $p$, then the $12n$ voters from $V_{pd}$ use up all
  their budget. The remaining voters who approve $d$, those in $U'$,
  have total budget equal to at most $3n \cdot \frac{1}{18n+6} < 1$,
  so $d$ cannot be selected in any of the following iterations (within
  Phase~1).  On the other hand, if we select $d$, then all the voters
  from $U'$ would have to pay all they had left (that is, either $0$
  or $\frac{1}{18n+6}$, each) and voters from $V_{pd}$ would split the
  remaining cost. That is, each voter from $V_{pd}$ would have to pay
  at least:
  \[
    \textstyle \frac{1 - 3n \cdot \frac{1}{18n+6}}{12n} = \frac{18n+6
      - 3n}{12n \cdot (18n+6)} = \frac{15n+6}{12n \cdot (18n+6)}.
  \]
  Consequently, each voter from $V_{pd}$ would be left with at most:
  \[ \textstyle
    \frac{1}{12n} - \frac{15n+6}{12n \cdot (18n+6)} = \frac{18n+6 -
      (15n+6)}{12n \cdot (18n+6)} = \frac{1}{72n+24}.
  \]
  This would not suffice to purchase $p$, as
  $12n \cdot \frac{1}{72n+24} < 1$.  Thus either we select $d$ (and
  not~$p$) or we select $p$ (and not $d$; where this is possible only
  if we previously purchased $S$-candidates that cover all members of
  $U'$).

  In the following iterations, we purchase all remaining
  $S$-candidates (because each of them is approved by three unique
  voters from $V_S$), as well as candidates $c_1$ and $c_2$ (voters
  $d_1$ and $d_2$ buy them with per-voter cost of $\nicefrac{1}{2}$
  for each). This uses up the budget of $d_1$ and, so, no candidate
  from $D$ is selected.  All in all, if there is no exact cover for
  the input \textsc{X3C} instance, then $W_d$ is the unique winning
  committee, but otherwise $W_d$ and $W_p$ tie. This finishes the
  proof.
\end{proof}

$\uniquecommittee$ is also $\conp$-complete for the full version of
MEqS. To see this, it suffices to note that after adding enough voters
with empty votes, MEqS becomes equivalent to Phragm{\'e}n (because
per-voter budget is so low that Phase~1 becomes vacuous) and inherits
its hardness.

% , not restricted to Phase~1. The problem
% remains $\conp$-complete.  Indeed, it suffices to note that if we
% extend an election with sufficiently many voters who do not approve
% any candidates, then the per-voter budget that MEqS assigns becomes so
% low that Phase~1 does not select any candidates. On such an election,
% MEqS is equivalent to Phragm{\'e}n and, consequently, inherits its
% $\conp$-hardness of $\uniquecommittee$ (membership in $\conp$ is
% clear).

On the positive side, for
sequential rules 
%MEqS, Phase~1 of MEqS, Phragm{\'e}n and the greedy
%rules % variants of Thiele rules
we can solve \textsc{Unique-Committee} %have an
in $\fpt$ time with respect to the committee size: In essence, we
first compute some winning committee and then we try all ways of
breaking internal ties to find a different one.  For small values of
$k$, such as, e.g., $k \leq 10$, the algorithm is fast enough to be
practical.

\begin{theorem}\label{thm:greedy-fpt-unique}
  Let $f$ be MEqS, Phase~1 of MEqS, Phragm{\'e}n, or a greedy variant
  of a $1$-concave Thiele rule. There is an $\fpt$ algorithm for
  $f$-$\uniquecommittee$ parameterized by the committee size.
\end{theorem}
\begin{proof}
  Let $E$ % = (C,V)$
  be the input election and let
  $k$ be the committee size. First, we compute some committee
  $W$ in $f(E,k)$, by running the algorithm for
  $f$ and breaking the internal ties arbitrarily. Next, we
  rerun %the following procedure:
%  \begin{enumerate}
%  \item[]
%  We run t
  the algorithm, but whenever it is about to add a candidate into the
  constructed committee, we do as follows (let
  $T$ be the set of candidates that the algorithm can insert into the
  committee): If $T$ contains some candidate
  $c$ that does not belong to
  $W$, then we halt and indicate that there are at least two winning
  committees ($W$ and those that include $c$). If
  $T$ is a subset of
  $W$, then we recursively try each way of breaking the tie. If the
  algorithm completes without halting, we report that there is a
  unique winning committee.
%  \end{enumerate}
  The correctness is immediate. The running time is equal to
  $O(k!)$ times the running time of the rule's algorithm (for the case
  where each tie is broken in a given way). Indeed, at the first
  internal tie we may need to recurse over at most
  $k$ different candidates, then over at most $k-1$, and so on.
\end{proof}

For $1$-concave Thiele rules other than
$\av$, \textsc{Unique-Committee} is
$\mathrm{co}\hbox{-}\wone$-hard when parameterized by the committee
size (this follows from the proof of
Proposition~\ref{pro:thiele-unique} as \textsc{Independent-Set} is
$\wone$-hard for parameter
$k$). To solve the problem in practice, we note that for each
$1$-concave Thiele rule there is an integer linear program (ILP) whose
solution corresponds to the winning committee.  We can either use the
ability of some ILP solvers to output several solutions (which only
succeeds in case of a tie), or we can use the following strategy:
First, we compute some winning committee using the basic ILP
formulation. Then, we extend the formulation with a constraint that
requires the committee to be different from the previous one and
compute a new one. If both committees have the same score, then there
is a tie.

% \textsc{Unique-Committee} for
% $1$-concave Thiele rules in practice, we
% express the winner determination 

% use their standard
% formulation as integer linear programs (ILPs); see, e.g., the textbook
% of \citet{lac-sko:b:multiwinner-approval}.  We can either use the
% ability of some ILP solvers that can output several solutions for a
% given ILP (which allows us to recognize ties), or we can use the
% following strategy: First, we compute some winning committee using the
% basic ILP formulation. Then, we extend the formulation with a
% constraint that requires the committee to be different from the
% previous one and compute a new one. If both comittees have the same
% score, then there is a tie.

%~\ref{pro:thiele-unique}, noting that
%\textsc{Independent-Set} is $\wone$-complete).

%For greedy variants of Thiele rules, as well as for
%Phragm{\'e}n, the complexity of $\uniquecommittee$ parameterized by the number of voters remains elusive (except for GreedyCCAV).

\section{Counting Winning Committees}\label{sec:counting}

Let us now consider the problem of counting the 
winning committees. Formally, our problem is as follows.

\begin{definition}
  Let $f$ be a multiwinner voting rule. In the
  $f$-\#\textsc{Winning-Committees} problem we are given an election
  and a committee size $k$; we ask for $|f(E,k)|$.
\end{definition}

There are polynomial-time algorithms for computing the number of
winning committees for $\av$ and $\sav$. For an election $E$ with
committee size $k$, we first sort the candidates with respect to their
scores in a non-increasing order and we let $x$ be the score of the
$k$-th candidate.  Then, we let $S$ be the number of candidates whose
score is greater than $x$, and we let $T$ be the number of candidates
with score equal to $x$. There are $T \choose k-S$ winning
committees. %The same reasoning works for SAV.

\begin{proposition}
  $\{$AV, \!\! SAV$\}$-\#\textsc{Winning-Committees} $\in\! \p$
\end{proposition}

On the other hand, whenever $f$-\textsc{Unique-Committee} is
intractable, so is $f$-\#\textsc{Winning-Committees}. Indeed, it
immediately follows that there is no polynomial-time
$(2-\varepsilon)$-approximation algorithm for
$f$-\#\textsc{Winning-Committees} for any $\varepsilon > 0$ (if such
an algorithm existed then it could solve $f$-\textsc{Unique-Committee} in
polynomial time as for an election with a single winning committee it
would have to output~$1$, and for an election with $2$ winning
committees or more, it would have to output an integer greater or
equal at least $\frac{2}{2-\varepsilon} > 1$, so we could distinguish
these cases\footnote{We assume here that if a solution for a counting
  problem is $x \in \mathbb{N}$, then an $\alpha$-approximation
  algorithm, with $\alpha \geq 1$, has to output an integer between
  $x/\alpha$ and $\alpha x$. If we allowed rational values on output,
  the inapproximability bound would drop to $\sqrt{2}-\varepsilon$.}).
%Hence, we have the following corollary.
%(we use the
%condition $\p = \np$ because it is more natural than the equivalent
%$\p = \conp$ one).
However, for all our rules a much stronger result holds.
\begin{proposition}\label{pro:hard-approx}
  Let $f$ be a $1$-concave Thiele rule (different from AV), its greedy
  variant, Phragm{\'e}n, MEqS or Phase~1 of MEqS.  Unless
  $\p \neq \np$, there is no polynomial-time approximation algorithm
  for $f$-\#\textsc{Winning-Committees} with polynomially-bounded
  approximation ratio.
\end{proposition}

\begin{proof}
    For Phase~1 of MEqS, it suffices to
    use the proof of Theorem~\ref{thm:meqs-p1-uni} with candidate $p$
    replaced by polynomially many copies, each approved by the same
    voters. Either we get a unique winning committee or polynomially many
    tied ones. The same trick works with 
    the greedy variants of $1$-concave Thiele rules and 
    Theorem~\ref{thm:uniqgreedyhard}, and
    Phragm{\'e}n and
    Corollary~\ref{cor:cc-pav-phragmen}.
        
    For the case of $1$-concave Thiele rules, we use the following
    strategy.  Let $p$ be some positive integer and let $(G,k)$ be an
    instance of \textsc{Independent-Set}, where $G$ is a graph and $k$ is
    an integer. We form a graph $G^p$ whose vertex set is:
    $V(G^p) = \{ v^i \mid v \in V(G), i \in [p] \}$ and where two
    vertices, $u^i$ and $v^j$, are connected by an edge either if
    $i \neq j$ or if $i=j$ and $u$ and $v$ are connected by an edge in
    $G$. Consequently, if $G$ has $x$ independent sets of size $k$, then
    $G^p$ has $px$ such sets (each independent set of $G^p$ is a copy of
    an independent set of $G$, using only vertices with the same
    superscript).  Hence, if in the proof of
    Propostion~\ref{pro:thiele-unique} we replace graph $G$ with graph
    $G^p$, where $p$ is some polynomial function of the input size, then
    we obtain an election that either has a unique winning committee (if
    the input graph did not have an independent set of a required size) or
    an election that has polynomially many winning committees (if the
    graph had at least one such independent set).
\end{proof}

% \begin{corollary}
%   Let $f$ be a $1$-concave Thiele rule ($f \neq$ AV).
% %  , its greedy variant, or
% %  Phragm{\'e}n.
%   For each positive $\varepsilon$, there is no polynomial-time
%   $(2-\varepsilon)$-approximation algorithm for
%   $f$-\#\textsc{Winning-Committees} unless $\p = \np$.
% \end{corollary}

% For Phase~1 of MEqS,
% %MEqS, its Phase~1, Phragm{\'e}n and greedy variants of $1$-concave
% %Thiele rules
% it suffices to use the proof of Theorem~\ref{thm:meqs-p1-uni} with
% polynomially many copies of $p$, each approved by the same voters as
% $p$: Either we get a unique solution or a tie among polynomially many
% committees.  The same trick works for the other sequential rules; for
% 1-concave Thiele rules we extend the proof of
% Proposition~\ref{pro:thiele-unique} with a trick that multiplies
% independent sets in the input graph.

% \begin{corollary}
%   Let $f$ be a greedy variant of a $1$-concave Thiele rule (different
%   from AV), Phragm{\'e}n, MEqS or Phase~1 of MEqS.  Unless
%   $\p \neq \np$, there is no polynomial-time approximation algorithm
%   for $f$-\#\textsc{Winning-Committees} with polynomially-bounded
%   approximation ratio.
% \end{corollary}

We note that the construction given in the proof of
Proposition~\ref{pro:thiele-unique} also shows that for each
$1$-concave Thiele rule $f \neq$ AV, $f$-\#\textsc{Winning-Committees} is both
$\sharpp$-hard and $\#\wone$-hard for parameterization by the
committee size (because this reduction produces elections that have
one more winning committee than the number of size-$k$ independent
sets in the input graph, and counting independent sets is both
$\sharpp$-complete and $\#\wone$-complete for parameterization by~$k$~\citep{val:j:permanent,flu-gro:j:parameterized-counting}). For
greedy variants of $1$-concave Thiele rules and Phragm{\'e}n, the
situation is more interesting because \textsc{Unique-Committee} is in
$\fpt$ (for the parameterization by the committee size). Yet,
\#\textsc{Winning-Committees} is also hard.

\begin{theorem}\label{thm:greedy-count}
  Let $f$ be Phragm{\'e}n or a greedy variant of a $1$-con\-cave Thiele rule (different
  from AV). $f$-\#\textsc{Winning-Committees} is $\sharpp$-hard and
  $\#\wone$-hard (for the parameterization by the committee size).
\end{theorem}
\begin{proof}
  We first consider greedy variants of $1$-concave Thiele rules.  Let
  $w$ be the weight function used by $f$. Let $x = w(2) - w(1)$. We
  have $w(1) = 1$ and we assume that $x < 1$ (we will consider the
  other case later).  We show a reduction from the $\sharpmatching$
  problem, where we are given a graph~$G$, an integer~$k$, and we ask
  for the number of size-$k$ matchings (i.e., the number of size-$k$
  sets of edges such that no two edges in the set share a
  vertex). $\sharpmatching$ is $\#\wone$-hard for parameterization by
  $k$~\citep{cur-mar:c:param-counting-matchings}.
    
  Let $G$ and $k$ be our input.  We form an election $E$ where the
  edges of $G$ are the candidates and the vertices are the voters. For
  each edge $e = \{u,v\}$, the corresponding edge candidate is
  approved by the vertex voters corresponding to $u$ and $v$. We also
  form an election $E_p$, equal to $E$ except that it has two
  extra voters who both approve a single new candidate, $p$.

  We note that every candidate in both $E$ and $E_p$ is approved by
  exactly two voters. Hence, the greedy procedure first keeps on
  choosing candidates whose score is $2$ (i.e., edges that jointly
  form a matching, or candidate $p$ in $E_p$). It selects the
  candidates with lower scores (i.e., edges that break a matching)
  only when score-$2$ candidates disappear.

  Let $W$ be some size-$k$ $f$-winning committee for election~$E_p$.
  We consider two cases:
  \begin{enumerate}
  \item If $p$ does not belong to $W$, then the edge candidates in~$W$
    form a matching. If it were not the case, then before including an
    edge candidate with score lower than $2$, the greedy algorithm
    would have included $p$ in the committee.

  \item If $p$ belongs to $W$ then $W \setminus \{p\}$ is an
    $f$-winning committee of size $k-1$ for election $E$. Indeed, if
    we take the run of the greedy algorithm that computes $W$ and
    remove the iteration where $p$ is selected, we get a correct run
    of the algorithm for election $E$ and committee size
    $k-1$. Further, for every size-$(k-1)$ committee winning in $E$,
    $S \cup \{p\}$ is a size-$k$ winning committee in $E_p$ (because
    we can always select $p$ in the first iteration).
  \end{enumerate}

  So, to compute the number of size-$k$ matchings in $G$, it suffices
  to count the number of winning size-$k$ committees in $E_p$ and subtract
  from it the
  number of winning size-$(k-1)$ committees in $E$.
%  , and output the
%  difference between the latter and the former.
%
    If $x = 1$, then we find the smallest value $t$ such that
    $w(t)-w(t-1) = 1$ and $w(t+1)-w(t) < 1$ and use the same construction
    as above, except that there are $t-1$ dummy candidates approved by
    every voter.
    
    Regarding Phragm{\'e}n, it turns out that the same construction as for
    the greedy variants of $1$-concave Thiele rules still works.  In time
    $t=\nicefrac{1}{2}$, each voter has $\nicefrac{1}{2}$ budget and each
    candidate (including $p$) can be purchased (because each candidate is
    approved by exactly two voters and their total budget is $1$). Hence,
    if $W$ is a winning committee for $E_p$ but $W$ does not include $p$,
    then all its members were purchased at time $\nicefrac{1}{2}$.
    %
    % Thus, since voters $x, y$ approve only $p$, $p$ does not belong
    % to a winning committee $W$ for $E_p$ if and only if its members can be
    % bought in time $\nicefrac{1}{2}$.
    %
    It means that these candidates were approved by disjoints sets of
    voters, whose corresponding edges edges form a size-$k$ matching.  On
    the other hand, if $p$ belongs to $W$ then $W \setminus {p}$ is a
    winning size-$(k-1)$ committee for $E$, as in the above proof.
\end{proof}

% the first iteration of Phragm{\'e}n rule and voters $x, y$ do not
% affect other voters.  Thus the number of $k$-size matchings in $G$ is
% equal to the difference between the number of $k$-size winning
% committees in $E_p$ and the number of winning $(k-1)$-size committees
% in $E$.

\begin{corollary}
  \#\textsc{Winning-Committees} is $\sharpp$-hard and $\#\wone$-hard (for the parameterization
  by the committee size) for GreedyCCAV, GreedyPAV, Phragm{\'e}n, and MEqS.
\end{corollary}

The above result holds for MEqS because of its relation to
Phragm{\'e}n. For Phase~1 of MEqS, we have $\sharpp$-hardness, but
$\sharpwone$-hardness so far remains elusive.

\begin{theorem}\label{thm:meqs-counting}
  \#\textsc{Winning-Committees} is $\sharpp$-hard for Phase~1 of MEqS.
\end{theorem}

\begin{proof}
    We give a reduction from \textsc{\#X3C}, i.e., a counting version of
    the problem used in the proof of Theorem~\ref{thm:meqs-p1-uni}.  Let
    $E_{pd}$ be the same election as constructed in that proof, except for
    the following change: Group $B_U$ contains $9n$ voters fewer and the
    $9n$ voters from $V_S$ additionally approve the candidates from
    $C_U$. Consequently, the committee size decreases by $9n$ (because we
    maintain that the committee size is equal to the number of voters).
    Because of this change, when selecting the candidates from $C_U$, the
    budget of the voters from $V_S$ drops to
    $\nicefrac{1}{(18n+6)}$. Then, after the iterations where
    $S$-candidates are selected at per-voter cost of
    $\nicefrac{1}{(18n+6)}$, no further $S$-candidates are selected
    (because the voters approving them have total budget lower than
    $1$). As a consequence, Phase~1 of MEqS applied to election $E_{pd}$
    chooses all committees of the following forms:
    \begin{enumerate}
    \item Committees consisting of all candidates from
      $C_B \cup C_U \cup \{c_1,c_2\} \cup \{d\}$ and a subset of
      $S$-candidates such that all other $S$-candidates include at least
      one covered voter from $U' \cup U''$.
    \item Committees consisting of all candidates from
      $C_B \cup C_U \cup \{c_1,c_2\} \cup \{p\}$ and a subset of
      $S$-candidates that correspond to an exact cover of $U$.
    \end{enumerate}
    Next, we form election $E_d$ identical to $E_{pd}$ except that it does
    not include candidate $p$. For $E_d$, Phase~1 of MEqS selects all the
    committees of the first type above.  Hence, to compute the number of
    solutions for our instance of \textsc{\#X3C}, it suffices to subtract
    the number of committees selected by Phase~1 of MEqS for $E_{d}$ from
    the number of committees selected by Phase~1 of MEqS for
    $E_{pd}$. This completes the proof.
\end{proof}

\section{Experiments}

A'priori, it is not clear how frequent are ties in multiwinner
elections. % or if it mostly is a theoretical possibility.
In this section we present experiments that show that they, indeed,
are quite common, at least if one considers %synthetic
elections of moderate size.

\subsection{Statistical Cultures and the Basic Experiment}

Below we describe the statistical cultures that we use to generate elections (namely, 
the resampling model, the interval model, and PabuLib data)
and how we perform our basic experiments.

% We focus on elections generated either according to the resampling
% model of
% \citet{szu-fal-jan-lac-sli-sor-tal:c:sampling-approval-elections}, or
% using a variant of the the well-known geometric model, or based on
% real-life data from PabuLib~\citep{sto-szu-tal:t:pabulib}.  Below we
% describe how we obtain our elections.

% following models (the
% first two are due to
% \citet{szu-fal-jan-lac-sli-sor-tal:c:sampling-approval-elections},
% whereas the third one is a well-known one).
% % \begin{enumerate}
% %\item

\paragraph{Resampling Model~\citep{szu-fal-jan-lac-sli-sor-tal:c:sampling-approval-elections}.}
We have two parameters, $p$ and $\phi$, both between $0$ and $1$. To
generate an election with candidate set $C = \{c_1, \ldots, c_m\}$ and
with $n$ voters, we first choose uniformly at random a central vote
$u$ approving exactly $\lfloor p m \rfloor$ candidates. Then, we
generate the votes, for each considering the candidates independently,
one by one. For a vote~$v$ and candidate~$c$, with probability
$1-\phi$ we copy $c$'s approval status from $u$ to $v$ (i.e., if $u$
approves $c$, then so does $v$; if $u$ does not approve $c$ then
neither does $v$), and with probability $\phi$ we ``resample'' the
approval status of~$c$, i.e., we let~$v$ approve~$c$ with probability
$p$ (and disapprove it with probability $1-p$).  On average, each
voter approves about $pm$ candidates.
% E.g.,
%if $\phi = 0$ then all the
%votes are identical and equal to the central, one and
% if $\phi = 1$ then each vote approves each candidate independently
% with probability $p$.

% \item
\paragraph{Interval Model.}
In the Interval model, each voter and each candidate is a point on a
$[0,1]$ interval, chosen uniformly at random. Additionally, each
candidate $c$ has radius $r_c$ and a voter $v$ approves canidate $c$
if the distance between their points is at most $r_c$. Intuitively,
the larger the radius, the more appealing is a given candidate.
We generate the radii of the candidates
by taking a base radius $r$ as input and, then, choosing each
candidates' radius from the normal distribution with mean $r$ and
standard deviation $\nicefrac{r}{2}$.  
Such spatial models are discussed in detail, e.g.,
by \citet{enelow1984spatial,enelow1990advances}. In the approval
setting, they were recently considered, e.g., by
\citet{bre-fal-kac-nie2019:experimental_ejr} and
\citet{god-bat-sko-fal:c:2d}.

\paragraph{PabuLib Data.} PabuLib is a library of real-life
participatory budgeting (PB) instances, mostly from Polish
cities~\citep{sto-szu-tal:t:pabulib}. A PB instance is a multiwinner
election where the candidates (referred to as projects) have costs and
the goal is to choose a ``committee'' of at most a given total
cost. We restrict our attention to instances from Warsaw, which use
approval voting, and we disregard the cost information (while this
makes our data less realistic, we are not aware of other sources of
real-life data for approval elections that would include sufficiently
large candidate and voter sets).
To generate an election with $m$ candidates and $n$ voters, we
randomly select a Warsaw PB instance, remove all but $m$ candidates
with the highest approval score, and randomly draw $n$ voters (with
repetition, restricting our attention only to voters who approve at
least one of the remaining candidates).
We consider 120 PB instances from Warsaw that include at least 30
candidates (each of them includes at least one thousand votes, usually a few thousand).

\paragraph{Basic Experiment.}
% We present two main experiments, each consisting of a series of basic
% ones that we now describe.
In a basic experiment we fix the number of candidates $m$, the
committee size~$k$, and a statistical culture.  Then, for each number
$n$ of voters between $20$ and $100$ (with a step of $1$) we generate
$1000$ elections with $m$ candidates and $n$~voters, and for each of
them compute whether our rules have a unique winning committee (we
omit GreedyCCAV).
% (we omit GreedyCCAV due to long computation times).
Then we present a figure
that on the $x$ axis has the number of voters and on the $y$ axis has
the fraction of elections that had a unique winning committee for a
given rule.
For AV and SAV, we use the algorithm from the beginning of
Section~\ref{sec:counting}, for sequential rules we use the FPT
algorithm from Theorem~\ref{thm:greedy-fpt-unique}, and for CCAV and
PAV we use the ILP-based approach, with a solver that
provides multiple solutions.

% For CCAV and PAV, we use the ability of the \texttt{abcvoting}
% library~\citep{abcvoting} to provide several winning committees, if
% they exist (the library expresses the winner determination process as
% an ILP instance and requests a pool of solutions from the ILP solver).
% For the greedy rules and Phragm{\'e}n, we use the algorithm from
% Theorem~\ref{thm:greedy-fpt-unique}.

\subsection{Results}

\newcommand{\resampling}[3]{png_ijcai23/unique/resampling/m_#1_k_#2_reps_1000_params_('p', #3) ('phi', 0.75).png}

\newcommand{\pabulib}[2]{png_ijcai23/unique/pabulib_with_replacement/m_#1_k_#2_reps_1000_params_best_cands_num_#1.png}

\newcommand{\interval}[3]{png_ijcai23/unique/euclidean_cr/m_#1_k_#2_reps_1000_params_('radius', #3) ('dim', 1) ('space', 'uniform').png}

\begin{figure*}[t]
  \centering
  \begin{subfigure}{0.3\textwidth}
    \centering
    \includegraphics[width=4.5cm]{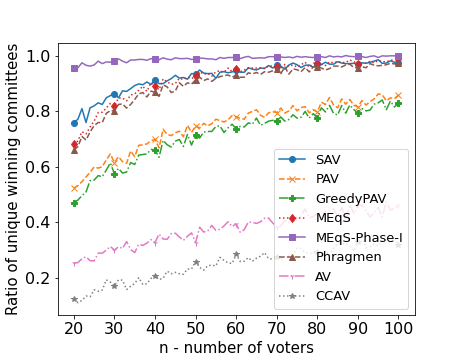}
    \caption{$m=30$, $k=5$, \\$k/2$ approvals/vote, \\ resampling model, $\phi = 0.75$}
  \end{subfigure}  
  \begin{subfigure}{0.3\textwidth}
    \centering
    \includegraphics[width=4.5cm]{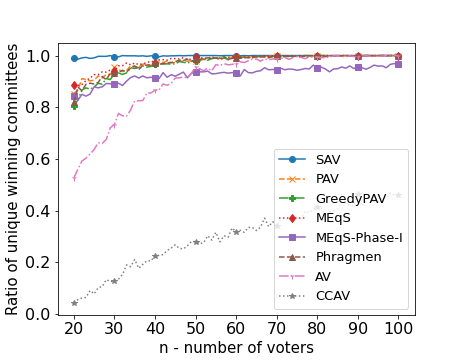}
    \caption{$m=30$, $k=5$, \\$k$ approvals/vote\\ resampling model, $\phi = 0.75$}
  \end{subfigure}  
  \begin{subfigure}{0.3\textwidth}
    \centering
    \includegraphics[width=4.5cm]{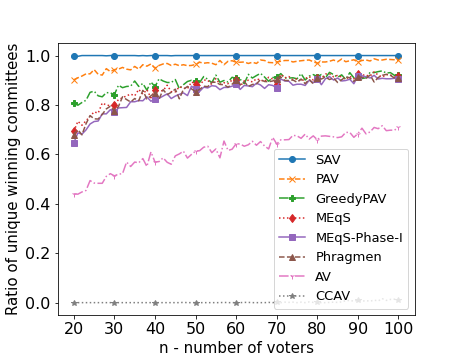}
    \caption{$m=30$, $k=5$, \\$2k$ approvals/vote \\ resampling model, $\phi = 0.75$}
  \end{subfigure}

  \begin{subfigure}{0.3\textwidth}
    \centering
    \includegraphics[width=4.5cm]{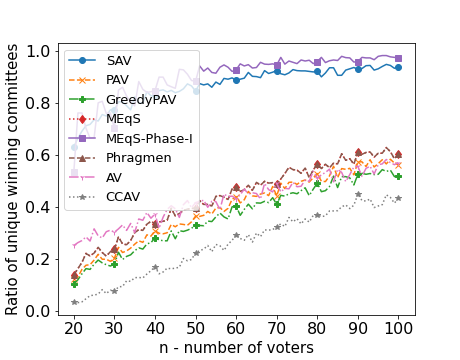}
    \caption{$m=30$, $k=5$, \\$k/2$ approvals/vote \\ Interval}
  \end{subfigure}  
  \begin{subfigure}{0.3\textwidth}
    \centering
    \includegraphics[width=4.5cm]{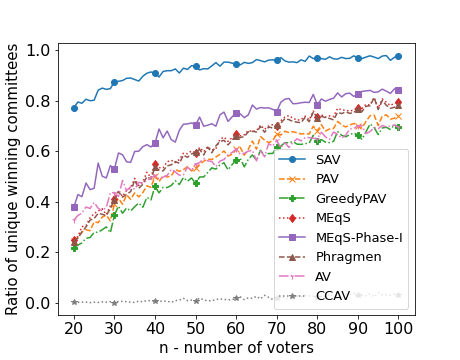}
    \caption{$m=30$, $k=5$, \\$k$ approvals/vote \\ Interval}
  \end{subfigure}  
  \begin{subfigure}{0.3\textwidth}
    \centering
    \includegraphics[width=4.5cm]{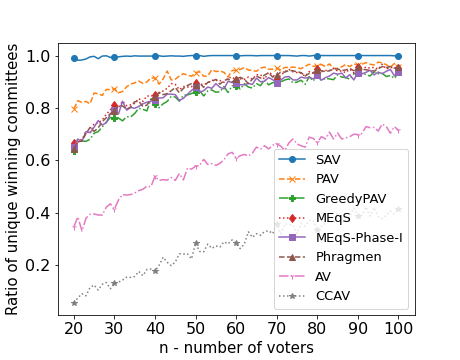}
    \caption{$m=30$, $k=5$,\\ PabuLib (Warsaw)\\}
  \end{subfigure}  

  % \begin{subfigure}{0.3\textwidth}
  %   \includegraphics[width=6.4cm]{\resampling{100}{10}{0.05}}
  %   \caption{$m=50$, $k=10$, $k/2$ approvals/vote}
  % \end{subfigure}  
  % \begin{subfigure}{0.3\textwidth}
  %   \includegraphics[width=6.4cm]{\resampling{100}{10}{0.1}}
  %   \caption{$m=50$, $k=10$, $k$ approvals/vote}
  % \end{subfigure}  
  % \begin{subfigure}{0.3\textwidth}
  %   \includegraphics[width=6.4cm]{\resampling{100}{10}{0.2}}
  %   \caption{$m=50$, $k=10$, $2k$ approvals/vote}
  % \end{subfigure}  

  \caption{\label{fig:resampling}Results of our experiments. By ``$k/2$
    approvals/vote'' we mean that on average a single vote contains
    approximately $k/2$ approvals (the meaning of $k$ and $2k$ is
    analogous).}
\end{figure*}

All our experiments regard $30$ candidates and committee size~$5$ (the
results for $50$ and $100$ candidates, and committee size $10$,
are analogous).  First, we performed three basic
experiments for the resampling model with the parameter $p$ (approval
probability) set so that, on average, each voter approved either
$k/2$, $k$, or $2k$ candidates. We used $\phi = 0.75$ (according to
the results of
\citet{szu-fal-jan-lac-sli-sor-tal:c:sampling-approval-elections},
this value gives elections that resemble the real-life ones).  We
present the results in the top row of
Figure~\ref{fig:resampling}. Next, we also performed two basic
experiments for the Interval model (with the base radius selected so
that, on average, each voter approved either $k/2$ or $k$ candidates),
and with the PabuLib data (see the second row of
Figure~\ref{fig:resampling}).  These experiments support the following
general conclusions.

First, for most scenarios and for most of our rules, there is a
nonnegligible probability of a tie (depending on the
rule and the number of voters, this probability may be as low as $5\%$
or as high as nearly $100\%$). This shows that one needs to be
ready to detect and handle ties in moderately sized multiwinner
elections.

Second, we see that SAV generally leads to fewest ties, CCAV leads to
most, and AV often holds a strong second position in this
category (in the sense that it also leads to a high probability of having a tie in many settings). The other rules are in between. Phase~1 of MEqS often has
significantly fewer ties than the other rules, but full version of
MEqS does not stand out. PAV occasionally leads to fewer ties (in
particular, on PabuLib data and on the resampling model with $2k$
approvals per vote).

\section{Summary}
We have shown that, in general, detecting ties in multiwinner
elections is intractable, but doing so for moderately-sized ones is
perfectly possible. Our experiments show that ties in such elections
are a realistic possibility and one should be ready to handle
them. Intractability of counting winning committees suggests that
tie-breaking by sampling committees may not be feasible. Looking for
fair tie-breaking mechanisms is a natural follow-up research direction.

\paragraph{Acknowledgments.}
This project has received funding from the European 
    Research Council (ERC) under the European Union’s Horizon 2020 
    research and innovation programme (grant agreement No 101002854).
    
    \noindent \includegraphics[width=3cm]{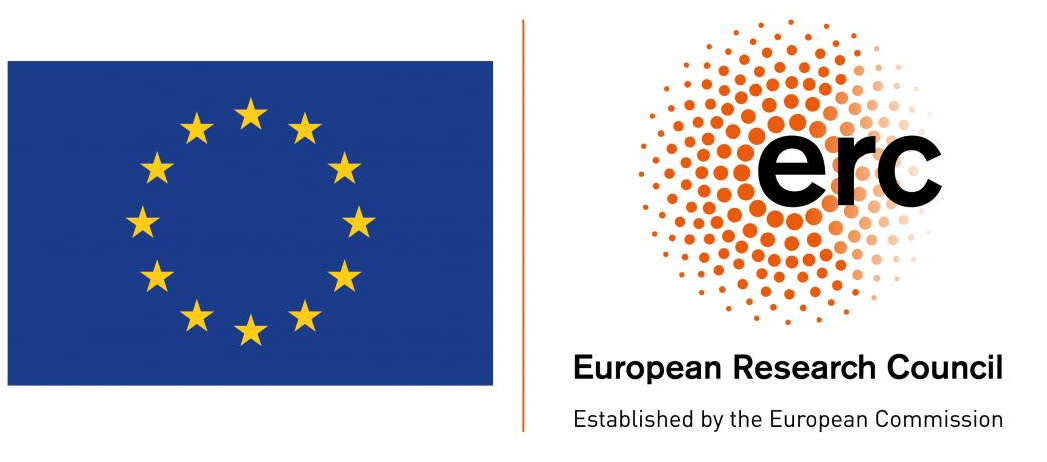}

\bibliographystyle{plainnat}
\bibliography{bib.bib}

\end{document}